\documentclass[11pt]{article}
\usepackage[utf8]{inputenc}

\usepackage{tikz}
\usetikzlibrary{positioning}
\usetikzlibrary{fit}

\usepackage{adjustbox}

\usepackage{amsmath}
\usepackage{amssymb}
\usepackage{amsthm}
\usepackage{authblk}

\usepackage{algpseudocode}
\usepackage{algorithm}
\usepackage{tikz}
\usetikzlibrary{positioning}

\usepackage{enumitem}

\usepackage{fullpage}

\usepackage{xcolor}

\usepackage{hyperref}

\usepackage[all]{hypcap}

\newcommand{\naturalnumber}{\ensuremath{{\mathbb{N}}}}
\newcommand{\naturalnumberpositive}{\ensuremath{{\mathbb{N}^+}}}

\newcommand{\p}{\ensuremath{{\rm P}}}
\newcommand{\np}{\ensuremath{{\rm NP}}}
\newcommand{\conp}{\ensuremath{{\rm coNP}}}
\newcommand{\ntime}{\ensuremath{{\rm NTIME}}}
\newcommand{\contime}{\ensuremath{\rm coNTIME}}

\title{A Critique of Lin's ``On $\np$ versus $\conp$ and Frege Systems''\thanks{Supported in part by NSF grant CCF-2006496}}

\author{Nicholas~DeJesse}
\author{Spencer~Lyudovyk}
\author{Dhruv~Pai}
\author{Michael~Reidy}

\affil{Department of Computer Science\\University of Rochester\\Rochester, NY 14627, USA}

\newtheorem{example}{Example}

\newtheorem{theorem}{Theorem}
\newtheorem{definition}{Definition}

\newtheorem{lemma}[theorem]{Lemma}

\date{May 8, 2025}

\begin{document}\sloppy

\maketitle

\begin{abstract}
In this paper, we examine Lin's ``On $\np$ versus $\conp$ and Frege Systems''~\cite{lin2024np}. Lin claims to prove that $\np~\neq~\conp$ by constructing a language $L_d$ such that $L_d \in \np$ but $L_d \notin \conp$. We present a flaw in Lin's construction of $D$ (a nondeterministic Turing machine that supposedly recognizes $L_d$ in polynomial time). We also provide a proof that $L_d \not\in \np$. In doing so, we demonstrate that Lin's claim that $\np\neq\conp$ is not established by his paper. In addition, we note that a number of further results that Lin claims are not validly established by his paper.
\end{abstract}

\section{Introduction}

This critique examines Lin's ``On $\np$ versus $\conp$ and Frege Systems''~\cite{lin2024np}, which claims to resolve the $\np$ versus $\conp$ problem by providing a language $L_d$ that is not in $\conp$ but is recognized by some nondeterministic Turing machine. Lin further elaborates on the implications of $\np \neq \conp$ by constructing results related to relativization techniques, intermediary languages, and Frege systems.

The $\np$ versus $\conp$ problem is incredibly important in the field of computational complexity. By proving that the two classes are equal, one would collapse the polynomial hierarchy to its second level. Conversely, if the two classes are not equal, then this would imply that $\p \neq \np$ since \mbox{P = coP} \cite[p.~322]{sip:b:introduction-third-edition}. Furthermore, $\np=\conp$ if and only if there exists a propositional proof system $\mathcal{P}$ (as defined by Cook and Reckhow in \cite{coo-rec:j:proof-systems}) such that all tautologies have polynomial-length proofs in $\mathcal{P}$. Thus, a resolution to this problem would have major consequences for the field.

In this paper, we argue that the reasoning Lin uses to prove his Theorem~$1$~\cite{lin2024np}, which states that $\np \neq \conp$, is unsound. Specifically, we present a major flaw in Lin's technique of proving that $L_d\in\np$, and we further prove that, in fact, $L_d \not\in$~NP\@. As a consequence, we show that the author's subsequent claims have not been established, as the author's proofs of these claims rely on Theorem~$1$.

\section{Preliminaries}
Let $\ntime[t(n)]$ denote the set of languages decidable by some nondeterministic Turing machine $M$ in time $O(t(n))$, as described in Sipser's textbook \cite{sip:b:introduction-third-edition}. Additionally, we assume the reader is familiar with Turing machines, nondeterminism, the complexity classes $\np$ and $\conp$, and big~O notation. For more information on any of these topics, readers can consult any standard textbook, e.g., \cite{aro-bar:computational-complexity, hop-ull:b:automata, sip:b:introduction-third-edition}.

Note that we are critiquing Version~$13$ of Lin's paper. We see that the paper is being frequently updated. However, our arguments are aimed at Lin's main claims, which have remained constant across all versions.

\subsection{The ``For All" Accepting Criterion and coNP}

In his paper, Lin defines the ``for all" accepting criterion for nondeterministic Turing machines as follows.
\begin{definition}[{\cite[Definition 2.4]{lin2024np}}] Let $M$ be a nondeterministic Turing machine. $M$ accepts the input $w$ if and only if all computation paths of $M(w)$ accept.
\end{definition}
This of course aims to capture ``conondeterministic" complexity, classes, and machines. Indeed, Lin provides the following definitions for $\contime$ and $\conp$ based on the ``for all" accepting criterion, which we have summarized for clarity.
\begin{definition}[{\cite[Definition 2.5]{lin2024np}}] Let $\contime [\mathit{t(n)}]$ denote the set of languages decidable by some nondeterministic Turing machine $M$ in time $O(t(n))$ using the ``for all" accepting criterion. Thus $\conp = \bigcup_{k \in \naturalnumberpositive} \contime [\mathit{n^k}]$.
\end{definition}
This definition of $\conp$ is less common than the ``for all" rejecting criterion,\footnote{The ``for all" rejecting criterion defines a $\conp$ machine as accepting an input if and only if every computation branch of the simulation rejects \cite[p.~56]{aro-bar:computational-complexity}.} but it is still valid since under this definition $\conp$ is still the set of all languages whose complement is in NP\@. More specifically, if a language $A$ is in $\conp$ under the ``for all" accepting criterion, then its complement $\overline{A}$ consists of all strings that have at least one branch of their computation reject when simulated on the Turing machine $M$ such that $L(M)=$~A\@. Therefore, flipping the accepting and rejecting states of $M$ causes it to accept a string only if at least one branch of its computation accepts, which is the definition of NP\@. This means $\overline{A}\in$~NP\@. Since the ``for all" accepting criterion is valid, this is the definition of $\conp$ we will use for the remainder of the paper.

\subsection{Universal Turing Machines}

In his paper, Lin uses a universal nondeterministic Turing machine to construct a machine which recognizes a language that is crucial to his main claim that $\np \neq$~coNP\@. 
A universal polynomial-time nondeterministic Turing machine $U$ takes in an encoding $\langle M, w\rangle$ of a polynomial-time nondeterministic Turing machine $M$ and string $w$ as an input and simulates $M$ on $w$. Thus, $U( \langle M, w\rangle )$ accepts if and only if $M(w)$ accepts. The encoding $\langle M, w\rangle$ is typically encoded over $\{ 0, 1 \} ^*$ using a standard encoding of the tuple $\langle M, w \rangle$~\cite{aho-hop-ull:b:design-analysis}. Furthermore, if $\langle M, w\rangle$ is not a syntactically correct encoding (and cannot be decoded), $U$ must handle this input accordingly.

\section{Analysis of Lin's Arguments}

\subsection{Introduction to Lin's Arguments}
\label{s:intro-argument}

Lin begins his paper by presenting a method for enumerating all machines that recognize languages in coNP\@. Lin first encodes each nondeterministic Turing machine into some string $t\in\{0,1\}^\ast$ using common practices~\cite{aho-hop-ull:b:design-analysis}. Then Lin uses an enumeration $e$ to enumerate over all strings $s\in\{0,1\}^\ast$. For each string $s$ enumerated in this way, $e$ attempts to decode $s$ into a nondeterministic Turing machine. Any string that cannot be decoded in this way is said to represent the trivial nondeterministic Turing machine with no transitions in its transition function.

In order to show that $\np\neq\conp$, Lin presents a language $L_d$ recognized by a universal five-tape polynomial-time nondeterministic Turing machine $D$ \cite[p.~17]{lin2024np}. $D$ takes as input a string $x$ with length $|x|$ and operates as described in the proof of Theorem~$4.1$ of \cite[p.~16]{lin2024np}:
\begin{enumerate}
    \item Let $M_1,~M_2,~\ldots$ be an enumeration of the encodings of all $\conp$ machines (as described previously), and let $M_i$ be the $i$th encoding that is enumerated this way. If $x \notin 1^\ast\langle M_i\rangle$, then $D$ rejects $x$. Otherwise, $x = 1^l\langle M_i\rangle$ for some $l\in\naturalnumber$, so $D$ decodes $\langle M_i\rangle$ and retrieves the following information from it:
    \begin{itemize}
        \item $t$ = the number of tape symbols used by $\langle M_i\rangle$.
        \item $s$ = the number of states in $\langle M_i\rangle$.
        \item $k$ = the time bound of $\langle M_i\rangle$. In other words, $\langle M_i\rangle$ runs in at most $n^k + k$ steps.
        \item $m$ = the shortest length of $\langle M_i\rangle$, which can be obtained by taking the input $x$ and removing the $l$ leading 1's.
    \end{itemize}
    $D$ uses its fifth tape as scratch paper to perform these calculations.
    \item On its second tape, $D$ creates $|x|$ blocks of $\lceil\log t\rceil$ cells, where each block is delineated by the marker symbol $\#$. This process creates $(1 + \lceil\log t \rceil)|x|$ cells in total. This tape is used to simulate the tape of $M_i$, and during this simulation, each block stores a binary-encoded version of the symbol in the corresponding cell of $M_i$'s tape. $D$ initializes this tape to store its input (encoded in binary as described above), encoding the binary representation for the blank cell in any unused blocks.
    \item On its third tape, $D$ creates a block of $\lceil (k+1)\log n\rceil$ cells, each initialized to $0$. $D$ uses this tape to count up to $n^{k+1}$, where $n = |x|$. Each transition of $M_i$ increments the counter stored on tape $3$, and if the tape ever overflows (i.e., if the count ever exceeds $n^{k+1}$), $D$ halts and rejects.
    \item On input $1^{n-m} \langle M_i \rangle$, $D$ computes an integer $j$ such that $f(j) < n \leq f(j+1)$, where $f(j+1)$ is defined recursively as
    \begin{equation*}
        f(j+1)=
        \begin{cases}
            2, &j=0\\
            2^{f(j)^k}, &j\geq 1.
        \end{cases}
    \end{equation*}
    Note that, by Lin's Remark~$4.1$ \cite[p.~15]{lin2024np}, this calculation can be done in polynomial time with respect to $|x|$. Then:
    \begin{enumerate}
        \item If $f(j) < n < f(j+1)$, then $D$ simulates $M_i$ on input $1^{n+1-m} \langle M_i\rangle$ using nondeterminism in $(n + 1)^{k} + k$ time and outputs the resulting answer.
        \item If $n = f(j+1)$, then $D$ simulates $M_i$ on input $1^{f(j)+1-m} \langle M_i\rangle$ and rejects $1^{n - m} \langle M_i\rangle$ if $M_i$ accepts $1^{f(j)+1-m} \langle M_i\rangle$ in $(f(j) + 1)^{k} + k$ time and accepts $1^{n - m} \langle M_i\rangle$ if $M_i$ rejects $1^{f(j)+1-m} \langle M_i\rangle$ in $(f(j) + 1)^{k} + k$ time.
    \end{enumerate}
\end{enumerate}

In its construction, $D$ uses a technique called lazy diagonalization (sometimes referred to as delayed diagonalization) to ensure that $D$ has enough time to simulate $M_i$ deterministically. The general idea behind this technique is to delay the diagonalization step for a sufficiently long, exponential number of steps until the machine can diagonalize against some part of its previous history (Ladner~\cite{lad:j:np-incomplete}, see also~\cite{hem-spa:j:team-diagonalization}). In this case, by the time $M_i$ is simulated deterministically in step $4b$, the input is already exponentially larger than the number of branches in $M_i$ because of step $4a$ \cite{lin2024np}. This means the deterministic simulation of $M_i$ can be performed in polynomial time with respect to the size of the input.

Lin then uses a proof by contradiction to claim that $L_d\notin$~coNP\@. Supposing that $L_d\in\conp$, Lin first shows that $D$ has enough time to simulate $M_i$, where $M_i$ is a $\conp$-machine such that $L(M_i) = L_d$, given some sufficiently long input $w$ containing the encoding of $M_i$ \cite{lin2024np}. Then, Lin demonstrates that
\begin{enumerate}
    \item [(1)] If $f(j) < n < f(j+1)$, then $D(1^{n-m}\langle M_i\rangle )=M_i(1^{n+1-m}\langle M_i\rangle )$ \cite[Eq.~4.1]{lin2024np}.
    \item [(2)] $D(1^{f(j+1)-m}\langle M_i\rangle )\neq M_i(1^{f(j)+1-m}\langle M_i\rangle )$ \cite[Eq.~4.2]{lin2024np}.
\end{enumerate}
which both follow easily from parts $4a$ and $4b$ of the definition of $D$. Specifically, part $4a$ of the definition of $D$ states that if $f(j) < n < f(j+1)$, then $D$ simulates $M_i$ on input $1^{n+1-m}\langle M_i\rangle$ and accepts if and only if $M_i$ accepts, meaning that the outputs of the machines $D$ and $M_i$ are equal on these inputs and Lin's Equation $4.1$ holds. Similarly, part $4b$ of the definition of $D$ states that if $n = f(j+1)$, then $D$ simulates $M_i$ on input $1^{f(j)+1-m}$ and outputs the opposite result, meaning that the outputs of the machines $D$ and $M_i$ are not equal on these inputs and Lin's Equation $4.2$ holds. Note that equality in this context means that the machines $D$ and $M_i$ have the same output (i.e. either both accept or both reject when given their respective inputs). Similarly, inequality means that $D$ and $M_i$ have different outputs (i.e., one accepts and the other rejects) on their respective inputs. Lin's Equation~$4.1$ yields
\begin{gather*}
    D(1^{f(j)+1-m}\langle M_i\rangle) = M_i(1^{f(j)+2-m}\langle M_i\rangle) \\
    D(1^{f(j)+2-m}\langle M_i\rangle) = M_i(1^{f(j)+3-m}\langle M_i\rangle) \\
    \vdots \\
    D(1^{f(j+1)-1-m}\langle M_i\rangle) = M_i(1^{f(j+1)-m}\langle M_i\rangle).
\end{gather*}

In all, this yields
\begin{align*}
    D(1^{f(j+1)-m}\langle M_i\rangle) = M_i(1^{f(j)+1-m}\langle M_i\rangle).
\end{align*}

However, this directly contradicts Lin's Equation~$4.2$~\cite{lin2024np}. Therefore, the assumption that $L_d\in\conp$ must be false, so $L_d\notin\conp$. We found no errors in Lin's enumeration over the $\conp$-machines, construction of $L_d$, or proof that $L_d\notin\conp$.

\subsection{Analysis of Lin's Theorem 5.3}
\label{s:main-claim}

In this section, we look at Lin's Theorem~$5.3$~\cite{lin2024np}, where the language $L_d^i$ is introduced as the set of strings that are accepted by the previously defined universal polynomial-time nondeterministic Turing machine $D$ within $O(n^i)$ steps for some $i \in \naturalnumberpositive$ \cite[p.~32]{lin2024np}.

The author then makes the following three claims:
\begin{enumerate}
    \item $L_d=\bigcup_{i\in\naturalnumberpositive}L_d^i$ {\cite[Eq.~5.1]{lin2024np}}
    \item For each $i \in \naturalnumberpositive$, $L_d^i \subseteq L_d^{i+1}$ {\cite[Eq.~5.2]{lin2024np}}
    \item $L_d^i \in \ntime[n^i] \subseteq \np$ for each $i\in \naturalnumberpositive$ {\cite[Eq.~5.3]{lin2024np}}
\end{enumerate}

Lin states that these three claims ``easily yield" $L_d \in \np$, resulting in his Theorem~$1$~\cite{lin2024np}.

\begin{theorem}[{\cite[Theorem 5.3]{lin2024np}}]
The language $L_d$ is in $\np$, where $L_d$ is accepted by the universal nondeterministic Turing machine $D$ that for some $k \in \naturalnumberpositive$ runs within time $O(n^k)$.
\end{theorem}

Although the author does not provide reasoning as to why $L_d \in \np$ follows from the three claims above, we believe that the intended intuition is as follows: $L_d \in \np$ because one can construct $L_d$ as the union of all $L_d^i$ for $i~\in~\naturalnumberpositive$, since every string in the language $L_d$ is guaranteed to be accepted by $D$ in some finite number of steps \cite[Eq.~5.1]{lin2024np}. Note that Lin's Equation~5.2 yields $L_d^1 \subseteq L_d^2 \subseteq L_d^3\subseteq\ldots$ {\cite[Eq.~5.2]{lin2024np}}, meaning each $L_d^{i+1}$ is a language containing at least all the strings in $L_d^i$. Therefore, one can construct a nondeterministic Turing machine that accepts $L_d$ by simulating the universal nondeterministic Turing machine $D$ for $i$ steps to obtain $L_d^i$ for each $i \in \naturalnumberpositive$ and then taking the union of each of these languages to get $L_d$. Furthermore, since each string in $L_d^i$ can by definition be accepted by $D$ in $O(n^i)$ steps, $L_d^i \in \ntime[n^i] \subseteq \np$ \cite[Eq.~5.3]{lin2024np}, which implies that $L_d \in \np$ since the class $\np$ is closed under union, meaning if some languages $L_1, L_2 \in \np$, then $L_1 \cup L_2 \in \np$.

We can easily see that $\np$ is closed under finite union by considering some languages $L_1, L_2 \in \np$ (see \cite[p.~322]{sip:b:introduction-third-edition}). By definition, there must be some nondeterministic polynomial-time Turing machines $N_1$ and $N_2$ which decide $L_1$ and $L_2$, respectively. Then, we can construct a new machine $N_3$ which decides $L_1 \cup L_2$ as follows:

\paragraph{\textnormal{$N_3$ = ``On input $w$:}}
\begin{enumerate}
    \item Run $N_1$ on $w$.
    \item If $N_1$ accepts, then accept $w$. Otherwise, continue.
    \item Run $N_2$ on $w$.
    \item If $N_2$ accepts, then accept $w$. Otherwise, reject $w$."
\end{enumerate}

If $w \in L_1 \cup L_2$, then one of $N_1$ or $N_2$ will accept it, so $N_3$ will accept it. If $w \not\in L_1 \cup L_2$, then both $N_1$ and $N_2$ will reject $w$, so $N_3$ will reject it. Thus, $N_3$ accepts $w$ if and only if $w \in L_1 \cup L_2$.

Then, since $L_1$ and $L_2$ run in nondeterministic polynomial time, suppose $L_1$ runs in $O(n^i)$ and $L_2$ runs in $O(n^j)$ for some $i,j \in \naturalnumberpositive$\@. Then, step~$(1)$ of $N_3$ runs in $O(n^i)$, step~$(3)$ runs in $O(n^j)$, and steps~$(2)$~and~$(4)$ run in $O(1)$, so $N_3$ runs in $O(n^i) + O(1) + O(n^j) + O(1) = O(n^{\max(i, j)})$ nondeterministic time. Thus, $N_3$ runs in nondeterministic polynomial time, so $L_1 \cup L_2 \in \np$. It thus follows that $\np$ is closed under finite union.

However, Lin seems to mistakenly believe that $\np$ is necessarily closed under infinite unions, which would imply that the union of infinitely many $\np$ languages, as in \cite[Eq.~5.1]{lin2024np}, is in NP\@. Every language is an infinite union of 1-element $\np$ (specifically constant-time) languages, namely of each of its elements. Thus, if $\np$ is closed under infinite unions, then every language is in $\np$, which is false. Thus constructing $L_d$ as Lin describes does not necessarily result in $L_d \in \np$.

In fact, we can even prove that $L_{d} \notin \np$, thus directly contradicting Lin's claim that $L_{d} \in \np$. 
To do this, we first prove two lemmas about Lin's function $f$ (which is described in our Section~\ref{s:intro-argument}) that will be useful in showing $L_d \notin \np$.

\begin{lemma}
\label{l:positive-increasing-lemma}
    For all positive integers $j$, $f(j)$ and $f(j + 1)$ are even positive integers such that $f(j) < f(j + 1)$.
\end{lemma}

\begin{lemma}
\label{l:diff-by-2-lemma}
    For all integers $j \geq 1$, $f(j) < f(j+1)-2$.
\end{lemma}

The proofs of these lemmas are deferred to Appendix~\ref{appendix:extra}. Using these two lemmas, we show that Lin's claim of $L_d\in\np$ does not hold.

\begin{theorem}
\label{t:disprove-main-claim}
    $L_{d} \notin \np$.
\end{theorem}

\begin{proof}
    Suppose for the sake of contradiction that $L_{d} \in \np$. This means that $\overline{L_{d}} \in \conp$, so there exists some $\conp$-machine $M_{i}$ with encoding $1^\ast\langle M_i \rangle$ in Lin's enumeration $e$ (see our Section~\ref{s:intro-argument}) such that $\langle M_i\rangle$ is the encoding of $M_i$ and $L(M_{i}) = \overline{L_{d}}$. Let $m = |\langle M_i\rangle|$. Since $m$ is a fixed nonnegative integer, it follows that there exists some integer $j>1$ such that $f(j) \geq m$.
    
    Since $M_{i}$ is a $\conp$-machine, it follows by part $4a$ of the definition of $D$ that for all nonnegative integers $n$ such that $f(j) < n < f(j + 1)$, $D$ accepts $1^{n - m} \langle M_i \rangle$ if and only if $M_{i}$ accepts $1^{n + 1 - m} \langle M_i \rangle$. Next, since $L(M_{i}) = \overline{L_{d}}$ and $L(D) = L_{d}$ (note that this means $L(M_{i}) = \overline{L(D)}$), then for all nonnegative integers $n$ such that $f(j) < n < f(j + 1)$, $M_i$ accepts $1^{n+1-m}\langle M_i \rangle$ if and only if $D$ rejects $1^{n+1-m}\langle M_i \rangle$. Thus, it follows that $D$ accepts $1^{n - m} \langle M_i \rangle$ if and only if $D$ rejects $1^{n + 1 - m} \langle M_i \rangle$ for all $n$ on the same interval. Therefore, for all nonnegative integers $n$ such that $f(j) < n < f(j + 1) - 1$, $D$ accepts $1^{n - m} \langle M_i \rangle$ if and only if $D$ accepts $1^{n + 2 - m} \langle M_i \rangle$, meaning $D(1^{n - m} \langle M_i \rangle) = D(1^{n + 2 - m} \langle M_i \rangle)$.
    
    Now, since $j > 1$, then by Lemma~\ref{l:diff-by-2-lemma}, we know $f(j) < f(j+1)-2$. Thus, since $f(j)$ is even by Lemma~\ref{l:positive-increasing-lemma}, then $f(j)+1$ is odd, and $f(j) < f(j)+1 < f(j+1)-2+1 = f(j+1)-1$. This means that $f(j)+1$ is odd, and $f(j) < f(j)+1 < f(j+1)-2+1 = f(j+1)-1$. Therefore, $f(j)+1$ is an odd integer between $f(j)$ and $f(j+1)-1$, exclusive. Thus, there exists some $n$ such that $f(j) < n < f(j+1)-1$. It follows that
    \begin{gather*}
        D(1^{f(j) + 1 - m} \langle M_i \rangle) = D(1^{f(j) + 3 - m} \langle M_i \rangle)\\
        D(1^{f(j) + 3 - m} \langle M_i \rangle) = D(1^{f(j) + 5 - m} \langle M_i \rangle)\\
        \vdots\\
        D(1^{f(j + 1) - 5 - m} \langle M_i \rangle) = D(1^{f(j + 1) - 3 - m} \langle M_i \rangle)\\
        D(1^{f(j + 1) - 3 - m} \langle M_i \rangle) = D(1^{f(j + 1) - 1 - m} \langle M_i \rangle).
    \end{gather*}

    In particular, since $f(j)$ and $f(j+1)$ are even and $f(j)+1$ is odd with $f(j)+1 < f(j+1)$, then there is some odd positive integer $k$ such that $f(j+1)-k = f(j)+1$. Then, as shown above, $D(1^{f(j+1)-1-m} \langle M_i \rangle) = D(1^{f(j+1)-3-m} \langle M_i \rangle) = \cdots = D(1^{f(j+1)-k-m} \langle M_i \rangle)$. Furthermore, since $f(j+1)-k = f(j)+1$, then $D(1^{f(j+1)-k-m} \langle M_i \rangle) = D(1^{f(j)+1-m} \langle M_i \rangle)$. It thus follows that $D(1^{f(j+1)-1-m} \langle M_i \rangle) = D(1^{f(j)+1-m} \langle M_i \rangle)$.
    
    Also, because $D$ accepts $1^{n - m} \langle M_i \rangle$ if and only if $D$ rejects $1^{n + 1 - m} \langle M_i \rangle$ for all nonnegative integers $n$ such that $f(j) < n < f(j + 1)$, we have

    \begin{equation*}
        D(1^{f(j + 1) - 1 - m} \langle M_i \rangle)\neq D(1^{f(j + 1) - m} \langle M_i \rangle).
    \end{equation*}
    
    When combined with the previous equations, this yields 

    \begin{equation*}
        D(1^{f(j) + 1 - m} \langle M_i \rangle) \neq D(1^{f(j + 1) - m} \langle M_i \rangle).
    \end{equation*}
    
    Next, by part $4b$ of the definition of $D$, $D$ accepts $1^{f(j + 1) - m} \langle M_i \rangle$ if and only if $M_i$ rejects $1^{f(j) + 1 - m} \langle M_i \rangle$. Then, since we established earlier that $L(M_{i}) = \overline{L(D)}$, it follows that $D$ accepts $1^{f(j) + 1 - m} \langle M_i \rangle$ if and only if $M_i$ rejects $1^{f(j) + 1 - m} \langle M_i \rangle$, which occurs if and only if $D$ accepts $1^{f(j+1) - m} \langle M_i \rangle$. Therefore, $D(1^{f(j) + 1 - m} \langle M_i \rangle) = D(1^{f(j + 1) - m} \langle M_i \rangle)$. However, this contradicts the fact that $D(1^{f(j) + 1 - m} \langle M_i \rangle) \neq D(1^{f(j + 1) - m} \langle M_i \rangle)$, as established earlier. Therefore, $L_{d} \notin \np$.
\end{proof}

\subsection{Analysis of Lin's Further Claims}

In this section, we discuss Sections $6$, $7$, and $8$ of Lin's paper, in which Lin claims to prove several results relating to his main theorems. However, since all of these results hinge on $\np \neq \conp$ via his Theorem~$1$, it is clear that they have not been resolved by Lin's paper~\cite{lin2024np}.

\subsubsection{Analysis of Lin's ``Breaking the Relativization Barrier"}

In Section~$6$, Lin introduces the idea of a co-nondeterministic oracle Turing machine, which we have summarized for clarity.
\begin{definition}[{\cite[Definition 6.2]{lin2024np}}]
    A co-nondeterministic oracle Turing machine $M$ is a nondeterministic Turing machine that has a special read-write tape, called an oracle tape, and three special states: $q_\text{query}$, $q_\text{yes}$, $q_\text{no}$. When constructing $M$, we specify some language $X\subseteq \Sigma ^\ast$ to act as the oracle language of $M$. Whenever $M$ enters the state $q_\text{query}$, the machine reads the contents of the oracle tape $w$, then moves into either $q_\text{yes}$ if $w\in X$ or into $q_\text{no}$ if $w \notin X$ in one computation step. We denote the co-nondeterministic oracle Turing machine $M$ on input $x$ with an oracle to the language $X$ as $M^X(x)$. $M^X$ is said to accept an input $x$ if all computation paths on input $x$ lead to an accepting state.
\end{definition}
Then Lin proposes three self-proclaimed ``rational assumptions" as follows:
\begin{enumerate}
    \item The polynomial-time co-nondeterministic oracle Turing machine can be encoded to a string over $\{ 0, 1 \}$,
    \item There are universal nondeterministic oracle Turing machines that can simulate any other co-nondeterministic oracle Turing machine, and
    \item The simulation can be done in time $O(T(n) \log T(n))$, where $T(n)$ is the time complexity of the simulated co-nondeterministic oracle Turing machine.
\end{enumerate}

Lin uses these assumptions for his Theorem~$2$, as shown below:

\begin{theorem}[{\cite[Theorem 2]{lin2024np}}]
    Under some rational assumptions, and if $\np^A = \conp^A$, then the set $\conp^A$ of all polynomial-time co-nondeterministic oracle Turing machine[s] with oracle $A$ is not enumerable. Thereby, the ordinary diagonalization techniques (lazy-diagonalization) will generally not apply to the relativized versions of the $\np$ versus $\conp$ problem.
\end{theorem}

Lin justifies this by contradiction. He argues that if the set $\conp^A$ of all polynomial-time co-nondeterministic oracle Turing machines is enumerable, then there exists a language $L_d^A$, constructed similarly to the language $L_d$ in Section 4 of his paper, such that $L_d^A \in \np^A$, but $L_d^A \notin \conp^A$. He states that this follows from the same logic he uses to prove that $L_d \in \np$ and $L_d \notin \conp$. Therefore, it follows that $\np^A \neq \conp^A$, which contradicts the assumption that $\np^A = \conp^A$. Thus, the set $\conp^A$ of all polynomial-time co-nondeterministic oracle Turing machines with oracle $A$ is not enumerable.

However, since in our Theorem~\ref{t:disprove-main-claim} we showed that Lin's proof that $L_d \in \np$ is flawed, the logic he uses in it cannot be used to show $L_d^A \in \np^A$. Thus, Lin's conclusion about the relativized version of the $\np$ versus $\conp$ problem is not validly established by his paper.

\subsubsection{Analysis of Lin's ``Rich Structure of coNP"}

In Section~$7$, Lin discusses what he calls the ``rich structure of $\conp$." He mentions Ladner's theorem, which states that if $\p \neq \np$, then there exist $\np$-intermediate problems, which are in $\np$ but are neither in $\p$ nor $\np$-complete \cite{lad:j:np-incomplete}. Then, based on his result that $\np \neq \conp$ in Theorem~$1$, Lin makes a similar claim regarding $\conp$ in Theorem~$6$.

\begin{theorem}[{\cite[Theorem 6]{lin2024np}}]
There are $\conp$-intermediate languages, i.e., there exists a language $L \in \conp$ that is neither in $\p$ nor in $\conp$-complete.
\end{theorem}

To prove this, Lin uses the fact that since $\np \neq \conp$ according to his Theorem~$1$, it follows that $\p \neq \np$, and thus, Ladner's $\np$-intermediate language $L_{\text{inter}}$ exists \cite{lad:j:np-incomplete}. He then defines a language $L$ to be the complement of $L_{\text{inter}}$, so $L \in \conp$. He argues that if $L \in \p$, then $L_{\text{inter}} \in \p$ since $\p$ is closed under complement \cite{sip:b:introduction-third-edition}, and if $L$ is $\conp$-complete, then $L_{\text{inter}}$ must be $\np$-complete \cite{hop-ull:b:automata}. Both of these cases contradict the definition of $L_{\text{inter}}$, so $L \not\in \p$ and $L$ is not $\conp$-complete, meaning $L$ must be $\conp$-intermediate.

However, this proof of the existence of such a $\conp$-intermediate language $L$ relies on Lin's Theorem~$1$. Thus, while we believe his conclusion that a $\conp$-intermediate language exists does indeed follow from $\np \neq \conp$, such a conclusion cannot be drawn since Lin has not demonstrated that $\np \neq \conp$.

\subsubsection{Analysis of Lin's ``Frege Systems"}
In Section~$8$, Lin introduces an application of his results to Frege proof systems, which are proof systems for propositional logic.

A proof system for a language $\mathcal{L} \subseteq \Sigma^{*}$ is a function $f : \Sigma_{1}^{*} \to \mathcal{L}$ for some alphabet $\Sigma_{1}$ that is polynomial-time computable and onto \cite{coo-rec:j:proof-systems}. Informally, the proof system defines what is considered to be a valid proof of a formula in $\mathcal{L}$, so that a formula has a valid proof if and only if it is in $\mathcal{L}$, and there is a polynomial-time algorithm for determining whether a given proof is valid for a given formula in $\mathcal{L}$. The proof system is considered to be polynomially bounded if there is a polynomial $p(n)$ such that the length of a proof $x \in \Sigma_{1}^{*}$ for any formula $y \in \mathcal{L}$ is bounded by $p(n)$ such that $|x| \leq p(|y|)$, where $|x|$ and $|y|$ are the number of symbols in $x$ and $y$, respectively \cite{coo-rec:j:proof-systems}.

Then, a Frege proof system is a proof system for propositional logic based on the language $\mathcal{L}$ of well-formed formulas constructed from a countable set of propositional variables and a finite propositionally complete set of connectives, a finite set of axiom schemes, and the inference rule of Modus Ponens. There exists such a Frege proof system for \rm{TAUT}, the language of all tautologies \cite{coo-rec:j:proof-systems}.

Lin cites Theorem~$8.1$ from \cite{coo-rec:j:proof-systems}, which states that any Frege proof system $f_1: \Sigma_1^* \to \mathcal{L}$ can $p$-simulate any other Frege proof system $f_2: \Sigma_2^* \to \mathcal{L}$, meaning there is a polynomial-time computable function $g: \Sigma_2^* \to \Sigma_1^*$ such that $f_1(g(x)) = f_2(x)$ for all $x \in \Sigma_2^*$. Thus, one Frege proof system is polynomially bounded if and only if all Frege proof systems are. Lin also cites Theorem~$8.2$ from \cite{coo-rec:j:proof-systems}, which states $\np$ is closed under complement if and only if \rm{TAUT} is in $\np$.

Using Theorems~$8.1$ and $8.2$ of $\cite{coo-rec:j:proof-systems}$, along with his Theorem~$1$, Lin claims to prove his Theorem~$8$, which is as follows:

\begin{theorem}[{\cite[Theorem 8]{lin2024np}}]
There exists no polynomial $p(n)$ such that for all $\psi \in \rm{TAUT}$, there is a Frege proof of $\psi$ of length at most $p(|\psi|)$. In other words, no Frege proof systems of $\psi \in \rm{TAUT}$ can be polynomially bounded.
\end{theorem}

Lin proves this by contradiction, stating that if there were a polynomially bounded proof system for \rm{TAUT} such that for some polynomial $p(n)$, for every formula $\phi \in \rm{TAUT}$, there exists a proof of length at most $p(|\phi|)$, then there would exist a nondeterministic Turing machine $M$ which guesses a proof of length at most $|\phi|$ and verifies it in polynomial time, thus implying $\rm{TAUT} \in \np$. Since \rm{TAUT} is $\conp$-complete (see~\cite[pg.~57]{aro-bar:computational-complexity}), this would imply $\np = \conp$, which contradicts Lin's Theorem~$1$. Thus, Lin concludes that no such polynomially-bounded proof system for \rm{TAUT} can exist. Additionally, since any Frege proof system can $p$-simulate another, then no Frege proof system can be polynomially bounded.

However, since our Theorem~\ref{t:disprove-main-claim} shows that Lin's Theorem~$1$ does not necessarily hold, it does not follow that $\np = \conp$ is false. Therefore, the conclusion that no polynomially-bounded proof system for \rm{TAUT} exists cannot be drawn.

\section{Conclusion}
In this paper, we have demonstrated that Lin's Theorem~$5.3$~\cite{lin2024np}, which implies that $\np \neq \conp$, is not established by Lin's paper. In particular, we have shown that Lin's proof of Theorem~$5.3$~\cite{lin2024np} relies both on his incorrect belief that the infinite union of $\np$ languages must be in $\np$ and on a different claim that also can be disproven. Thus the validity of Lin's Theorem~$1$ cannot be established as it relies on his Theorem~$5.3$. We have also shown that, as a result, the subsequent claims Lin makes also cannot be verified due to their reliance on his main claim in Theorem~$1$~\cite{lin2024np}.

\paragraph{Acknowledgments}
We would like to thank Lane~A.~Hemaspaandra, Harry~Liuson, \mbox{Isabel~Humphreys}, Matthew~Iceland, Dylan~McKellips, and Leo~Sciortino for their helpful comments on prior drafts. The authors are responsible for any remaining errors.

\bibliographystyle{alpha}
\bibliography{main}

\clearpage
\appendix
\section{Additional Proofs}
\label{appendix:extra}

Included below are the proofs of Lemmas~\ref{l:positive-increasing-lemma} and~\ref{l:diff-by-2-lemma} that are introduced in Section~\ref{s:main-claim}.

\begin{lemma}
    For all positive integers $j$, $f(j)$ and $f(j + 1)$ are even positive integers such that $f(j) < f(j + 1)$.
\end{lemma}

\begin{proof}
We proceed by induction over $j$.
First, fix some $k \in \naturalnumberpositive$ and consider $j=0$. Since $f(j+1) = 2^{f(j)^k}$ for some $k \geq 1$, clearly $f(0) = 2$ and $f(1) = 2^{f(0)^k}=2^{2^k}$, which are both even integers. It follows that since $2 < 2^{2^k}$, we must have $f(0) < f(1)$. Next, assume this holds for $j = 0, 1, \dots, n-1$. By definition $f(n) = 2 ^ {f(n-1) ^ k}$ and $f(n+1) = 2^{f(n)^k}$. Since $f(n-1) < f(n)$ by assumption, it must be the case that $2^{f(n-1)^k} < 2^{f(n)^k}$ implies $ f(n) < f(n+1)$. Furthermore, these terms must be even integers, as they are equivalent to two raised to some exponent.
Thus, for all positive integers $j$, $f(j) < f(j+1)$.
\end{proof}

\begin{lemma}
    For all integers $j \geq 1$, $f(j) < f(j+1)-2$.
\end{lemma}

\begin{proof}
    We will proceed by induction over $j$.
    First, consider $j = 2$. 
    Recall that $f(1) = 2$ and $f(2) = 2^{f(1)^k}$ for some $k\in\naturalnumberpositive$. Thus, $f(2) \geq 2^{2^1} = 2^2 = 4$. Note that for the same $k$, $f(3) = 2^{f(2)^k}$, so $f(3) \geq 2^{4^1} = 2^4 = 16$. Therefore, since $4 < 14$, it follows that $f(2) < f(3)-2$.

    Next, assume that this holds for all $j = 2, 3, \dots, n-1$, we will show that it also holds for $n$.
     Recall that $f(n-1) = 2^{f(n-2)^k}$ and $f(n) = 2^{f(n-1)^k}$ for some $k \in \naturalnumberpositive$; in particular, note that $k \geq 1$. Thus, 
     \begin{equation}
         f(n) = 2^{f(n-1)^k} > 2^{(1+f(n-2))^k} \geq 2^{1^k + f(n-2)^k} = 2^{1^k} \cdot 2^{f(n-2)^k} = 2 f(n-1).
     \end{equation}
     However, it is clear that $f(n) \geq f(2) = 4$, as by Lemma~\ref{l:positive-increasing-lemma} $f(p+1) > f(p)$ for all $p \in \naturalnumberpositive$. It follows that $2 f(n) = f(n) + f(n) \geq f(n) + 4 > f(n) + 2$. Therefore, because $f(n) > 2 f(n-1)$ and $2 f(n-1) > f(n-1)+2$, it must be the case that $f(n) > f(n-1) + 2$, or $f(n) < f(n-1)-2$.
\end{proof}

\end{document}